\documentclass{llncs}
  \usepackage{amssymb,amsmath}
  \usepackage{enumerate}
  \usepackage{graphicx}
  \usepackage{verbatim}

  \newcommand{\per}{\textsf{per}}
  
  \newcommand{\runs}{\rho}
  \newcommand{\cubicruns}{\rho_{cubic}}
  \newcommand{\rexp}{\sigma}
  \newcommand{\cubicrexp}{\sigma_{cubic}}

  \def\rdots{\mathinner{\ldotp\ldotp}}

  \date{}
  \author{\bf
    Maxime Crochemore\inst{1}\fnmsep\inst{3}
    \and
    Marcin Kubica\inst{2}
    \and
    Jakub Radoszewski%
    \thanks{%
      Some parts of this paper were written during the author's Erasmus exchange
      at King's College London
    }
    \inst{2}
    \and \\
    Wojciech Rytter\inst{2}\fnmsep\inst{5}
    \and
    Tomasz Wale\'n\inst{2}
  }

  \institute{
    King's College London, London WC2R 2LS, UK \\
    \email{maxime.crochemore@kcl.ac.uk}
    \and
    Dept.~of Mathematics, Computer Science and Mechanics, \\
    University of Warsaw, Warsaw, Poland\\
    \email{[kubica,jrad,rytter,walen]@mimuw.edu.pl}
    \and
    Universit\'e Paris-Est, France
    \and
    Dept. of Math. and Informatics,\\
    Copernicus University, Toru\'n, Poland
  }

  \title{
    On the Maximal Sum of Exponents\\ of Runs in a String
  }

\begin{document}
  \maketitle
  \begin{abstract}
    A run is an inclusion maximal occurrence in a string (as a subinterval)
    of a repetition $v$ with a period $p$ such that $2p \le |v|$.
    The exponent of a run is defined as $|v|/p$ and is $\ge 2$.
    We show new bounds on the maximal sum of exponents of runs in
    a string of length $n$.
    Our upper bound of $4.1\,n$ is better than the best previously known
    proven bound of $5.6\,n$ by Crochemore \& Ilie (2008).
    The lower bound of $2.035\,n$, obtained using a family of binary words,
    contradicts the conjecture of Kolpakov \& Kucherov (1999)
    that the maximal sum of exponents of runs in a string of length $n$
    is smaller than $2n$.
  \end{abstract}

  \section{Introduction}
    Repetitions and periodicities in strings are one of the fundamental topics in
    combinatorics on words \cite{Karhumaki,Lothaire}.
    They are also important in other areas: lossless compression, word representation, computational biology, etc.
    In this paper we consider bounds on the sum of exponents of repetitions that a string
    of a given length may contain.
    In general, repetitions are studied also from other points of view, like:
    the classification of words (both finite and infinite) not containing repetitions of a given exponent,
    efficient identification of factors being repetitions of different types
    and computing the bounds on the number of various types of
    repetitions occurring in a string.
    The known results in the topic and a deeper description of the
    motivation can be found in a survey by Crochemore et al.~\cite{Survey}.

    The concept of runs (also called maximal repetitions) has been introduced to
    represent all repetitions in a string in a succinct manner.
    The crucial property of runs is that their maximal number in a string of
    length $n$ (denoted as $\runs(n)$) is $O(n)$, see Kolpakov \& Kucherov \cite{KolpakovKucherov}.
    This fact is the cornerstone of any algorithm computing all repetitions in
    strings of length $n$ in $O(n)$ time.
    Due to the work of many people, much better bounds on $\runs(n)$ have been obtained.
    The lower bound $0.927\, n$ was first proved by Franek \& Yang \cite{Franek08}.
    Afterwards, it was improved by Kusano et al.~\cite{Matsubara} to $0.944565\, n$ employing computer experiments, 
    and very recently by Simpson~\cite{Simpson10} to $0.944575712\, n$.
    On the other hand, the first explicit upper bound $5\,n$ was settled by Rytter~\cite{Rytter06}, 
    afterwards it was systematically improved to $3.48\,n$ by Puglisi et al.~\cite{Puglisi08},
    $3.44\, n$ by Rytter~\cite{Rytter07},
    $1.6\, n$ by Crochemore \& Ilie~\cite{CrochemoreIlie,Crochemore08} and $1.52\, n$ by Giraud~\cite{Giraud08}.
    The best known result $\runs(n) \le 1.029\, n$ is due to Crochemore et 
    al.~\cite{DBLP:conf/cpm/CrochemoreIT08}, but it is conjectured
    \cite{KolpakovKucherov} that $\runs(n)<n$.
    Some results are known also for repetitions of exponent higher than 2.
    For instance, the maximal number of cubic runs (maximal repetitions with exponent at least 3)
    in a string of length $n$ (denoted $\cubicruns(n)$) is known to be between $0.406\,n$ and $0.5\,n$,
    see Crochemore et al.~\cite{Lata10}.

    A stronger property of runs is that the maximal sum of their exponents in a string
    of length $n$ (notation: $\rexp(n)$) is linear in terms of $n$, see Kolpakov \& Kucherov~\cite{KolpakovKucherovLORIA}.
    It has applications to the analysis of various algorithms, such as
    computing branching tandem repeats: the linearity of the sum of exponents
    solves a conjecture of \cite{Gusfield98} concerning the linearity of the number of maximal
    tandem repeats and implies that all can be found in linear time.
    For other applications, we refer to \cite{KolpakovKucherovLORIA}.
    The proof that $\rexp(n) < cn$ in Kolpakov and Kucherov's paper \cite{KolpakovKucherovLORIA} is very complex
    and does not provide any particular value for the constant $c$.
    A bound can be derived from the proof of Rytter \cite{Rytter06} but he mentioned only
    that the bound that he obtains is ``unsatisfactory'' (it seems to be $25\,n$).
    The first explicit bound $5.6\,n$ for $\rexp(n)$ was provided by Crochemore and Ilie \cite{Crochemore08},
    who claim that it could be improved to $2.9\,n$ employing computer experiments.
    As for the lower bound on $\rexp(n)$, no exact values were previously known and
    it was conjectured \cite{Kolpakov99,KolpakovKucherovLORIA} that $\rexp(n) < 2n$.

    In this paper we provide an upper bound of $4.1\,n$ on the maximal sum of exponents
    of runs in a string of length $n$ and also a stronger upper bound of $2.5\,n$ for the
    maximal sum of exponents of cubic runs in a string of length $n$.
    As for the lower bound, we bring down the conjecture $\rexp(n) < 2n$ by
    providing an infinite family of binary strings for which the sum of exponents of runs is greater than
    $2.035\,n$.

  \section{Preliminaries}
    We consider \emph{words} (\emph{strings}) $u$ over a finite alphabet $\Sigma$, $u \in \Sigma^*$;
    the empty word is denoted by $\varepsilon$;
    the positions in $u$ are numbered from $1$ to $|u|$.
    For $u=u_1u_2\ldots u_m$, let us denote by $u[i \rdots j]$ a \textit{factor}
    of $u$ equal to $u_i\ldots u_j$ (in particular $u[i]=u[i \rdots i]$).
    Words $u[1 \rdots i]$ are called prefixes of $u$, and words $u[i \rdots |u|]$ suffixes of $u$.

    We say that an integer $p$ is the (shortest) \emph{period} of a word
    $u=u_1\ldots u_m$ (notation: $p=\per(u)$) if $p$ is the smallest positive integer such
    that $u_i=u_{i+p}$ holds for all $1\le i\le m-p$.
    We say that words $u$ and $v$ are cyclically equivalent (or that one of them is a cyclic
    rotation of the other) if $u=xy$ and $v=yx$ for some $x, y \in \Sigma^*$.

    A \emph{run} (also called a maximal repetition) in a string $u$ is an interval
    $[i\rdots j]$ such that:
    \begin{itemize}
      \item
        the period $p$ of the associated factor $u[i\rdots j]$ satisfies
        $2p \le j-i+1$, 
      \item 
        the interval cannot be extended to the right nor to the left, without violating the above property, 
        that is, $u[i-1] \ne u[i+p-1]$ and $u[j-p+1] \ne u[j+1]$.
    \end{itemize}
    A \emph{cubic run} is a run $[i\rdots j]$ for which the
    shortest period $p$ satisfies $3p\le j-i+1$.
    For simplicity, in the rest of the text we sometimes refer to runs and cubic runs as
    to occurrences of the corresponding factors of $u$.
    The (fractional) \emph{exponent} of a run is defined as $(j-i+1)/p$.

    For a given word $u \in \Sigma^*$, we introduce the following notation:
    \begin{itemize}
      \item $\runs(u)$ and $\cubicruns(u)$ are the numbers of runs and cubic
      runs in $u$ resp.
      \item $\rexp(u)$ and $\cubicrexp(u)$ are the sums of exponents of runs and cubic
      runs in $u$ resp.
    \end{itemize}
    For a non-negative integer $n$, we use the same notations $\runs(n)$, $\cubicruns(n)$,
    $\rexp(n)$ and $\cubicrexp(n)$ to denote the maximal value of the respective function
    for a word of length $n$.

  \section{Lower bound for $\rexp(n)$}
    Tables \ref{fig:franek} and \ref{fig:padovan}
    list the sums of exponents of runs for several words of two known families
    that contain very large number of runs: the words $x_i$ defined by
    Franek and Yang \cite{Franek08} (giving the lower bound $\runs(n) \ge 0.927\,n$,
    conjectured for some time to be optimal)
    and the modified Padovan words $y_i$ defined by Simpson \cite{Simpson10}
    (giving the best known lower bound $\runs(n) \ge 0.944575712\,n$).
    These values have been computed experimentally.
    They suggest that for the families of words $x_i$ and $y_i$ the maximal sum of
    exponents could be less than $2n$.

    \begin{table}
      \begin{center}
      \begin{tabular*}{0.6\textwidth}{@{\extracolsep{\fill}}|r|r|r|r|r|}
      \hline
$i$ & $|x_i|$ & $\runs(x_i)/|x_i|$ & $\rexp(x_i)$ & $\rexp(x_i)/|x_i|$ \\\hline
$1$ & $6$ & $0.3333$ & $4.00$ & $0.6667$ \\\hline
$2$ & $27$ & $0.7037$ & $39.18$ & $1.4510$ \\\hline
$3$ & $116$ & $0.8534$ & $209.70$ & $1.8078$ \\\hline
$4$ & $493$ & $0.9047$ & $954.27$ & $1.9356$ \\\hline
$5$ & $2090$ & $0.9206$ & $4130.66$ & $1.9764$ \\\hline
$6$ & $8855$ & $0.9252$ & $17608.48$ & $1.9885$ \\\hline
$7$ & $37512$ & $0.9266$ & $74723.85$ & $1.9920$ \\\hline
$8$ & $158905$ & $0.9269$ & $316690.85$ & $1.9930$ \\\hline
$9$ & $673134$ & $0.9270$ & $1341701.95$ & $1.9932$ \\\hline
      \end{tabular*}
      \end{center}
      \caption{\label{fig:franek}
        Number of runs and sum of exponents of runs in Franek \& Yang's \cite{Franek08} words $x_i$.
      }
    \end{table}

    \begin{table}
      \begin{center}
      \begin{tabular*}{0.6\textwidth}{@{\extracolsep{\fill}}|r|r|r|r|r|}
      \hline
$i$ & $|y_i|$ & $\runs(y_i)/|y_i|$ & $\rexp(y_i)$ & $\rexp(y_i)/|y_i|$ \\\hline
$4$ & $37$ & $0.7568$ & $57.98$ & $1.5671$ \\\hline
$8$ & $125$ & $0.8640$ & $225.75$ & $1.8060$ \\\hline
$12$ & $380$ & $0.9079$ & $726.66$ & $1.9123$ \\\hline
$16$ & $1172$ & $0.9309$ & $2303.21$ & $1.9652$ \\\hline
$20$ & $3609$ & $0.9396$ & $7165.93$ & $1.9856$ \\\hline
$24$ & $11114$ & $0.9427$ & $22148.78$ & $1.9929$ \\\hline
$28$ & $34227$ & $0.9439$ & $68307.62$ & $1.9957$ \\\hline
$32$ & $105405$ & $0.9443$ & $210467.18$ & $1.9967$ \\\hline
$36$ & $324605$ & $0.9445$ & $648270.74$ & $1.9971$ \\\hline
$40$ & $999652$ & $0.9445$ & $1996544.30$ & $1.9972$ \\\hline
      \end{tabular*}
      \end{center}
      \caption{\label{fig:padovan}
        Number of runs and sum of exponents of runs in Simpson's \cite{Simpson10} modified Padovan words $y_i$.
      }
    \end{table}

  We show, however, a lower bound for $\rexp(n)$ that is greater than $2n$.

  \begin{theorem}
    There are infinitely many binary strings $w$ such that
    $$\frac{\rexp(w)}{|w|} > 2.035.$$
  \end{theorem}

  \begin{proof}
    Let us define two morphisms $\phi:\{a,b,c\}\mapsto\{a,b,c\}$ and $\psi:\{a,b,c\}\mapsto\{0,1\}$ as follows:
    $$\phi(a)=baaba, \quad \phi(b)=ca, \quad \phi(c)=bca$$
    $$\psi(a)=01011, \quad \psi(b)=\psi(c)=01001011$$
    We define $w_i=\psi(\phi^i(a))$.
    Table~\ref{fig:lower} shows the sums of exponents of runs in words $w_i$, computed experimentally.

    Clearly, for any word $w=(w_8)^k$, $k \ge 1$, we have
    $$\frac{\rexp(w)}{|w|} > 2.035.$$

    \begin{table}
      \begin{center}
      \begin{tabular*}{0.5\textwidth}{@{\extracolsep{\fill}}|r|r|r|r|}
      \hline
$i$ & $|w_i|$ & $\rexp(w_i)$ & $\rexp(w_i)/|w_i|$ \\\hline
$1$ & $31$ & $47.10$ & $1.5194$ \\\hline
$2$ & $119$ & $222.26$ & $1.8677$ \\\hline
$3$ & $461$ & $911.68$ & $1.9776$ \\\hline
$4$ & $1751$ & $3533.34$ & $2.0179$ \\\hline
$5$ & $6647$ & $13498.20$ & $2.0307$ \\\hline
$6$ & $25205$ & $51264.37$ & $2.0339$ \\\hline
$7$ & $95567$ & $194470.30$ & $2.0349$ \\\hline
$8$ & $362327$ & $737393.11$ & $2.0352$ \\\hline
$9$ & $1373693$ & $2795792.39$ & $2.0352$ \\\hline
$10$ & $5208071$ & $10599765.15$ & $2.0353$ \\\hline
      \end{tabular*}
      \end{center}
      \caption{\label{fig:lower}
        Sums of exponents of runs in words $w_i$.
      }
    \end{table}

  \qed
  \end{proof}

  \section{Upper bounds for $\rexp(n)$ and $\cubicrexp(n)$}
    In this section we utilize the concept of \emph{handles} of runs as
    defined in \cite{Lata10}.
    The original definition refers only to cubic runs, but here we extend it also to ordinary runs.

    Let $u \in \Sigma^*$ be a word of length $n$.
    Let us denote by $P=\{p_1,p_2,\ldots,p_{n-1}\}$ the set of inter-positions in $u$ 
    that are located \emph{between} pairs of consecutive letters of $u$.
    We define a function $H$ assigning to each run $v$ in $u$ 
    a set of some inter-positions within $v$ (called later on \emph{handles}) --- 
    $H$ is a mapping from the set of runs occurring in $u$ to the set $2^P$ of subsets of $P$.
    Let $v$ be a run with period $p$ and let $w$ be the prefix of $v$ of length $p$.
    Let $w_{\min}$ and $w_{\max}$ be the minimal and maximal words (in lexicographical order)
    cyclically equivalent to $w$.
    $H(v)$ is defined as follows:
    \begin{enumerate}[a)]
      \item
      if $w_{\min}=w_{\max}$ then $H(v)$ contains all inter-positions
      within $v$,
      \item
      if $w_{\min} \ne w_{\max}$ then $H(v)$ contains inter-positions between
      consecutive occurrences of $w_{\min}$ in $v$ and between consecutive
      occurrences of $w_{\max}$ in $v$.
    \end{enumerate}
    Note that $H(v)$ can be empty for a non-cubic-run $v$.

    \begin{figure}[th]
    \begin{center}
      \includegraphics[width=7cm]{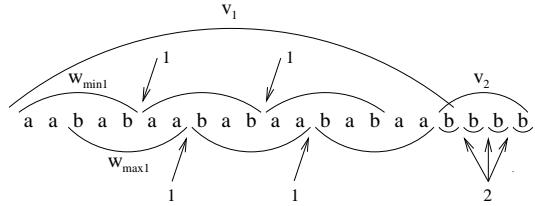}
      \caption{\label{f:handles_ex}
          An example of a word with two highlighted runs $v_1$ and $v_2$.
          For $v_1$ we have $w_{\mbox{\scriptsize min1}} \ne w_{\mbox{\scriptsize max1}}$ and for $v_2$
          the corresponding words are equal to $b$ (a one-letter word).
          The inter-positions belonging to the sets $H(v_1)$ and $H(v_2)$
          are pointed by arrows
      }
    \end{center}
    \end{figure}

    Proofs of the following properties of handles of runs can be found in \cite{Lata10}:
    \begin{enumerate}
      \item Case (a) in the definition of $H(v)$ implies that $|w_{\min}|=1$.
      \item $H(v_1) \cap H(v_2) = \emptyset$ for any two distinct runs $v_1$ and $v_2$ in $u$.
    \end{enumerate}

    To prove the upper bound for $\rexp(n)$, we need to state an additional property of
    handles of runs.
    Let $\mathcal{R}(u)$ be the set of all runs in a word $u$, and let
    $\mathcal{R}_1(u)$ and $\mathcal{R}_{\ge 2}(u)$ be the sets of runs with period
    1 and at least 2 respectively.

    \begin{lemma}
      \label{lem:propH}~

    \noindent
      If $v \in \mathcal{R}_1(u)$ then
      $\rexp(v) = |H(v)|+1$.\\
      If $v \in \mathcal{R}_{\ge 2}(u)$ then
      $\lceil\rexp(v)\rceil \le \frac{|H(v)|}2+3$.
    \end{lemma}

    \begin{proof}
      For the case of $v \in \mathcal{R}_1(u)$, the proof is straightforward from the definition of handles.
      In the opposite case, it is sufficient to note that both words
      $w_{\min}^k$ and $w_{\max}^k$ for $k=\lfloor\rexp(v)\rfloor-1$ are factors of $v$, and thus
      $$|H(v)| \ge 2\cdot(\lfloor\rexp(v)\rfloor-2).$$
      \qed
    \end{proof}

    Now we are ready to prove the upper bound for $\rexp(n)$.
    In the proof we use the bound $\runs(n) \le 1.029\,n$ on the number of runs from \cite{DBLP:conf/cpm/CrochemoreIT08}.

    \begin{theorem}\label{thm:upper_bound}
      The sum of the exponents of runs in a string of length $n$ is less than $4.1\,n$.
    \end{theorem}

    \begin{proof}
      Let $u$ be a word of length $n$.
      Using Lemma \ref{lem:propH}, we obtain:
      \begin{eqnarray}
      \nonumber
        \sum_{v \in \mathcal{R}(u)} \rexp(v) & = & \sum_{v \in \mathcal{R}_1(u)} \rexp(v) + \sum_{v \in \mathcal{R}_{\ge 2}(u)} \rexp(v) \\
      \nonumber
        & \le\ & \sum_{v \in \mathcal{R}_1(u)} \left(|H(v)|+1\right) + \sum_{v \in \mathcal{R}_{\ge 2}(u)} \left(\frac{|H(v)|}2+3\right) \\
      \nonumber
        & =\ & \sum_{v \in \mathcal{R}_1(u)} |H(v)| + |\mathcal{R}_1(u)| + \sum_{v \in \mathcal{R}_{\ge 2}(u)} \frac{|H(v)|}2 + 3\cdot|\mathcal{R}_{\ge 2}(u)| \\
      \label{eq:Ru}
        & \le\ & 3\cdot|\mathcal{R}(u)| + A + B/2,
      \end{eqnarray}
      where $A=\sum_{v \in \mathcal{R}_1(u)} |H(v)|$ and $B=\sum_{v \in \mathcal{R}_{\ge 2}(u)} |H(v)|$.
      Due to the disjointness of handles of runs (the second property of handles), $A+B<n$, and
      thus, $A+B/2<n$.
      Combining this with \eqref{eq:Ru}, we obtain:
      $$
        \sum_{v \in \mathcal{R}(u)} \rexp(v)\ <\ 3\cdot|\mathcal{R}(u)|+n\ \le\ 3\cdot\runs(n)+n\ \le\ 3\cdot1.029\,n+n\ <\ 4.1\,n.
      $$
      \qed
    \end{proof}

    A similar approach for cubic runs, this time using the bound of $0.5\,n$ for $\cubicruns(n)$ from \cite{Lata10},
    enables us to immediately provide a stronger upper bound for the function $\cubicrexp(n)$.

    \begin{theorem}
      The sum of the exponents of cubic runs in a string of length $n$ is less than $2.5\,n$.
    \end{theorem}

    \begin{proof}
      Let $u$ be a word of length $n$.
      Using same inequalities as in the proof of Theorem \ref{thm:upper_bound}, we obtain:
      $$\sum_{v \in \mathcal{R}_{cubic}(u)} \rexp(v) \ <\ 3\cdot|\mathcal{R}_{cubic}(u)|+n
        \ \le\ 3\cdot\cubicruns(n)+n\ \le\ 3\cdot0.5\,n+n\ =\ 2.5\,n,$$
      where $\mathcal{R}_{cubic}(u)$ denotes the set of all cubic runs of $u$.
      \qed
    \end{proof}

  \bibliographystyle{abbrv}
  \bibliography{exprun}

\end{document}